\documentclass[letterpaper, 10 pt, conference]{ieeeconf}

\IEEEoverridecommandlockouts                            
\usepackage{amsmath,amssymb,booktabs,colortbl,cases}
\usepackage{relsize,lipsum,tcolorbox}
\usepackage{dsfont}
\usepackage{url}
\usepackage[latin1]{inputenc}
\usepackage{eufrak}
\usepackage{epsfig,wrapfig}
\usepackage{color}
\usepackage[noadjust]{cite}
\usepackage{mathrsfs}
\usepackage{tikz,lipsum}
\usepackage{enumerate}
\usepackage{psfrag}
\usepackage{adjustbox,graphics,graphicx,subfig} 
\usetikzlibrary{arrows,shapes.geometric,positioning}
\usepackage{epsfig} 
\definecolor{my_GREEN}{rgb}{0,0.5,0}
\definecolor{LightCyan}{rgb}{0.95,0.95,0.9}
\definecolor{my_Gray}{rgb}{0.85,0.85,0.85}
\definecolor{myCOLOR1}{RGB}{0,0,255}

\newcommand{\editKG}[1]{{\textcolor{black}{#1}}}
\newcommand{\editKSG}[1]{{\textcolor{black}{#1}}}
\allowdisplaybreaks
\newtheorem{theorem}{Theorem}[section]

\newtheorem{proposition}[theorem]{Proposition}
\newtheorem{lemma}[theorem]{Lemma}
\newtheorem{definition}[theorem]{Definition}

\newtheorem{remark}[theorem]{Remark}
\DeclareMathOperator{\sgn}{sgn}

\makeatletter

\newcommand{\Rmnum}[1]{\expandafter\@slowromancap\romannumeral #1@}

\makeatother
\title{\LARGE \bf
Constant Bearing Pursuit on Branching Graphs
}
\author{Kevin S. Galloway$^{1}$ and Biswadip Dey$^{2}$
\thanks{*The second author's research was supported in part by the Office of Naval Research under ONR grant N00014-14-1-0635.}%
\thanks{$^{1}$Kevin S. Galloway is with the Electrical and Computer Engineering Department, United States Naval Academy, Annapolis, MD 21402, USA. {\tt\small kgallowa@usna.edu}}
\thanks{$^{2}$Biswadip Dey is with the Department of Mechanical and Aerospace Engineering, Princeton University, Princeton, NJ 08544, USA. {\tt\small biswadip@princeton.edu}}
}
\begin{document}
%
%
\maketitle
\thispagestyle{empty}
\pagestyle{empty}
%
%
\begin{abstract}
Cyclic pursuit frameworks provide an efficient way to create useful global behaviors out of pairwise interactions in a collective of autonomous robots. Earlier work studied cyclic pursuit with a constant bearing (CB) pursuit law, and has demonstrated the existence of a variety of interesting behaviors for the corresponding dynamics. In this work, by attaching multiple branches to a single cycle, we introduce a modified version of this framework which allows us to consider any weakly connected pursuit graph where each node has an outdegree of $1$. This provides a further generalization of the cyclic pursuit setting. Then, after showing existence of relative equilibria (rectilinear or circling motion), pure shape equilibria (spiraling motion) and periodic orbits, we also derive necessary conditions for stability of a 3-agent collective. By paving a way for individual agents to join or leave a collective without perturbing the motion of others, our approach leads to improved reliability of the overall system.
\end{abstract}
%

%
%
%
\section{Background}
\label{sec:intro}
%
Previous research \cite{Justh_PSK_SCL04, 5160735, CollectiveMot_PrstEscp_Couzin} has demonstrated the relevance and scope of pursuit interactions as a building block for collective behavior. In particular, cyclic pursuit (wherein agents pursue each other over a cycle graph) can be used to synthesize rectilinear, circling, and spiraling motions \cite{Marshall2004, Marshall20063, Sinha20071954, Galloway2013SymmetryStrategies, Galloway2016SymmetryPursuit}, which provide effective tools for a variety of missions, such as search-and-rescue and environmental sensing. A modified version of cyclic pursuit can also be employed in a beacon-referenced setting, wherein the beacon represents a target of interest for an unmanned vehicle or an attractive food source in a biological context \cite{Galloway2015StationPursuit, KSG_BD_2016_ACC, Mallik_Sinha_ECC_15, Daingade2016AImplementation}.

If each agent pursues exactly one other agent in a collective, and the pursuit graph is weakly connected, then the graph will contain a single cycle \cite{Galloway2013SymmetryStrategies}. While  previous work such as \cite{Marshall2004, Marshall20063, Sinha20071954, Galloway2013SymmetryStrategies, Galloway2016SymmetryPursuit} dealt with pursuit graphs which were only single cycles, in this current work we consider a more general case wherein the cycle may have \emph{branches} connected to it. We demonstrate that many of the same collective behaviors can be achieved (i.e. rectilinear and circling \editKG{motions}, as well as shape-preserving spirals), with several added benefits. 

\editKSG{One benefit of our proposed framework is that it} provides a straightforward method for independent agents to join (or leave) an existing collective without disturbing the other agents\editKG{. This also makes the collective more robust to unexpected agent losses, as compared to a collective composed of a single cycle, since branch agents can be lost without affecting the trajectories of other agents.} Similarly, our approach provides a method for agents to select and adjust their station within the collective, simply by adjusting the CB control parameter. This also provides a method for a small number of ``informed agents'' to establish an orbit or search pattern at a particular location, and a number of ``uninformed agents'' (e.g. no GPS) can be deployed to augment the coverage. Lastly, system stability can be analyzed in a tiered fashion, since agents in the branches are only influenced by other agents which are ``higher'' on the pursuit graph (in the sense of being closer to the cycle).  

The cycle-with-branches pursuit graph is the fundamental unit necessary to describe behavior in a population of agents which each pursue one other agent, since the pursuit graph of such a population can be represented by the union of a set of cycle-with-branches pursuit graphs which are each weakly connected. This pursuit graph can also be used to describe the behavior of ``remora'' agents which attach themselves to a collective, possibly for the purpose of identifying the states and control parameters employed by the agents in the collective, in the spirit of the work presented in \cite{Galloway2016StateSystems}.
%
%
%
%
%

%
%
%
\section{Modeling Pursuit Interactions}
\label{sec:model}
%
\subsection{Agents as Self-steering Particles}
%
As described in \cite{Justh_PSK_SCL04}, we view autonomous agents as self-steering particles moving on a plane. By letting ${\bf r}_i \in \mathds{R}^2$ denote the position and a unit vector ${\bf x}_{i}$ denote the normalized velocity of agent $i$, its dynamics can be expressed as
\begin{equation}
\dot{\bf r}_i  = {\bf x}_i,
\quad 
\dot{\bf x}_i  = u_i {\bf y}_i,
\quad \textrm{and} \quad 
\dot{\bf y}_i = -u_i {\bf x}_i,
\label{eqn:particleDynamics}
\end{equation}
where ${\bf y}_{i} ={\bf x}_{i}^{\perp}$ is obtained by rotating ${\bf x}_{i}$ by $\pi/2$ in the counter-clockwise direction and $u_i \in \mathds{R}$ denotes its steering (curvature) control. Here we have assumed that the agents travel at a common nonzero speed at any given point of time.

We consider a directed graph, i.e. the pursuit graph (alternatively known as the \emph{attention graph}), $ {\cal G} = ({\cal N},{\cal A})$ with node set $ {\cal N} = \{1,2,...,n\} $ and arc set $ {\cal A} $, where the nodes correspond to the individual agents and the arcs $(i,j) \in {\cal A} $ imply that agent $i$ is paying attention to agent $j$. The analysis in  \cite{Galloway2013SymmetryStrategies} demonstrated a symmetry reduction to \emph{shape space}, which removes the reference to an absolute coordinate frame and instead describes only the relative positions and velocities of the agents. The same work also introduced a polar parametrization which is particularly useful for describing the interactions between agents, which we review here. Letting $ R(\beta) $ denote the $ 2 \times 2 $ rotation matrix
\begin{align}
R(\beta) = \left[ \begin{array}{c c} \cos\beta & -\sin\beta \\ \sin\beta & \cos\beta \end{array}\right],
\end{align}
we define the shape variables $ \kappa_{ij} $, $ \theta_{ji} $, and $ \rho_{ij} $ by
\begin{equation}
\begin{aligned}
& 
R(\kappa_{ij}){\bf x}_i \cdot \frac{{\bf r}_j-{\bf r}_i}{|{\bf r}_j-{\bf r}_i|} = 1, 
\quad 
R(\theta_{ji}){\bf x}_j \cdot \frac{{\bf r}_i-{\bf r}_j}{|{\bf r}_i-{\bf r}_j|} = 1, 
\\
& 
\rho_{ji} = |{\bf r}_j-{\bf r}_i|,
\end{aligned}
\label{shapeVariableDefinitions}
\end{equation}
for $ (i,j) \in \cal{A} $, provided $ {\bf r}_i \ne {\bf r}_j $ (see Fig~\ref{shape_var_fig}). Using this parameterization, the shape dynamics can be expressed as
\begin{equation}
\begin{aligned}
\dot{\kappa}_{ij} &= -u_i + \frac{1}{\rho_{ij}}(\sin\kappa_{ij}+ \sin\theta_{ji}),
\\
\dot{\theta}_{ji} &= - u_j + \frac{1}{\rho_{ij}}(\sin\kappa_{ij}+ \sin\theta_{ji}),
\\
\dot{\rho}_{ij} &= -\cos\kappa_{ij}-\cos\theta_{ji},
\end{aligned}
\label{dynamics_ij}
\end{equation}
provided $ \rho_{ij} > 0 $. The requirement that $\rho_{ij}$ be strictly positive is a result of the polar singularity and is necessary to ensure that particular control laws we will define in follow-on sections are well-defined. Note however that the dynamics do not necessarily enforce this constraint. We note also that cyclic interactions incur an additional cycle closure constraint, which will be made explicit in what follows.
\begin{figure}
\centering
\includegraphics[width=.37\textwidth]{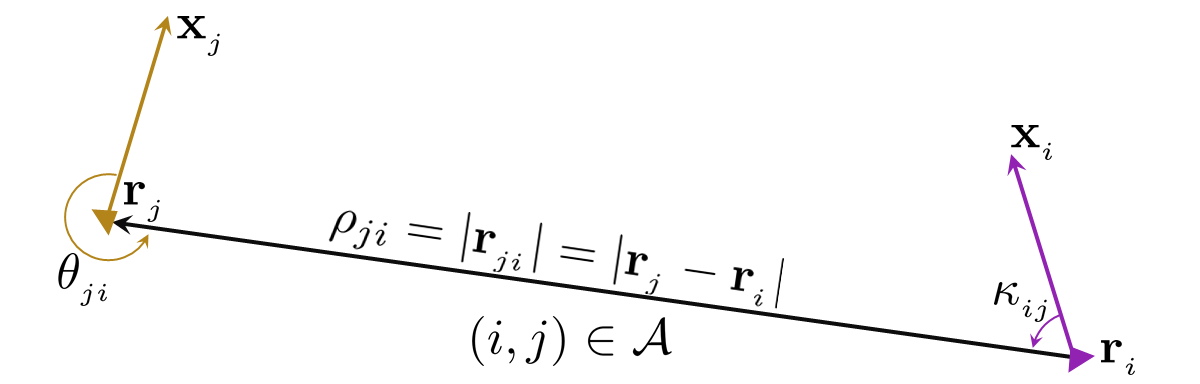}
\caption{\editKSG{Illustration} of scalar shape variables $\kappa_{ij}$, $\theta_{ji}$, $\rho_{ij}$.}
\label{shape_var_fig}
\vspace{-1.5em}
\end{figure}
%
%

%
%
\subsection{Constant Bearing Pursuit}
%
In the current work we consider the constant bearing (CB) pursuit strategy, which specifies that the pursuer should maneuver to maintain a specified angle $\alpha$ between its velocity vector and the line-of-sight to the pursuee. For a pursuer $i$ pursuing agent $j$, we can specify the CB pursuit strategy in terms of shape variables by defining 
\begin{align}
\label{cbcostshape}
\Lambda_i = -\cos(\kappa_{ij} - \alpha_i),
\end{align}
so that the CB pursuit strategy is attained if and only if $\Lambda_i = -1$. A feedback control law to achieve CB pursuit (originally developed in \cite{Wei2009PursuitGame}) is given by
\begin{align}
\label{ucbalphaishape}
u_{CB(\alpha_i)} = \mu_i \sin(\kappa_{ij}-\alpha_i) + \frac{1}{\rho_{ij}}\left(\sin\kappa_{ij}+\sin\theta_{ji}\right),
\end{align}
where $\mu_i>0$ is a control gain. If every agent $i$ in a given system pursues exactly one other agent $j$ with the CB pursuit law \eqref{ucbalphaishape} (with CB parameter $\alpha_i$), then we have the closed-loop shape dynamics
\begin{equation}
\begin{aligned}
\dot{\kappa}_{ij} &= -\mu_i \sin(\kappa_{ij}-\alpha_i), 
\\
\dot{\theta}_{ji} &= -\mu_j \sin(\kappa_{jk}-\alpha_j) -\frac{1}{\rho_{jk}}(\sin\kappa_{jk} + \sin\theta_{kj}) 
\\ 
& \qquad \quad + \frac{1}{\rho_{ij}}(\sin\kappa_{ij}+ \sin\theta_{ji}),
\\
\dot{\rho}_{ij} &= -\cos\kappa_{ij}-\cos\theta_{ji}, 
\end{aligned}
\label{closedloopcb}
\end{equation}
for $ i = 1,2,...,n $, where $ (i,j), (j,k) \in {\cal A} $.  

The dynamics \eqref{closedloopcb} possess the interesting characteristic of rendering a certain manifold invariant and attractive, in the following sense. If we define the \editKSG{CB pursuit manifold} $M_{CB(\pmb{\alpha)}}$ as
the set of all states for which each agent has attained the specified CB pursuit strategy (i.e. $\Lambda_i=-1$), trajectories which start on $M_{CB(\pmb{\alpha)}}$ will remain on the manifold for all future time, and under certain assumptions (see \cite{Galloway2013SymmetryStrategies} for more details) system trajectories which begin in a large subset of the shape space asymptotically approach $M_{CB(\pmb{\alpha)}}$ as $t \rightarrow \infty$. Thus we find it useful to consider the dynamics \eqref{closedloopcb} restricted to $M_{CB(\pmb{\alpha)}}$, which are given by
\begin{equation}
\begin{aligned}
\dot{\theta}_{ji} 
&= 
-\frac{1}{\rho_{jk}}(\sin\alpha_j + \sin\theta_{kj}) +  \frac{1}{\rho_{ij}}(\sin\alpha_i+ \sin\theta_{ji}),
\\
\dot{\rho}_{ij} 
&= 
-\cos\alpha_i- \cos\theta_{ji}.
\end{aligned}
\label{invarmandyn}
\end{equation}
%
%
%

%
%
%
\section{Cyclic pursuit with a single branch}
\label{sec:BranchDynamics}
%
In this section we consider a pursuit graph which consists of a cycle with a single ``branch''. More plainly, we consider the case in which agents $1$ through $n-1$ are engaged in cyclic CB pursuit, and agent $n$ (the branch) employs CB pursuit with regards to one of the agents on the cycle. (Without loss of generality, we assume that agent $n$ pursues agent $1$). This leads us to consider the dynamics for an $n$-agent system in which each agent employs the CB pursuit feedback law, and the arc set is given by 
\begin{equation}
\label{eqn:ArcSet}
{\cal A} = {\cal A}_{cycle} \cup \left\{(n,1) \right\},
\end{equation}
where ${\cal A}_{cycle} = \left\{ (1,2), (2,3), \ldots, (n-1,1) \right\}$.
Note that agent $n$ is influenced by the behavior of agents in the cycle but does not exert any influence on them.
%
%

%
%
\subsection{Shape dynamics and relative equilibria}
%
We consider the shape dynamics restricted to $M_{CB(\pmb{\alpha)}}$, i.e. our shape dynamics are given by \eqref{invarmandyn} for $i=1,2,\ldots, n-1$, with the branch dynamics given by
\begin{align}
\dot{\theta}_{1n} 
&= 
- \frac{1}{\rho_{12}}\left(\sin \alpha_{1} + \sin \theta_{21} \right) + \frac{1}{\rho_{n1}}\left(\sin \alpha_{n} + \sin \theta_{1n} \right),
\nonumber 
\\
\dot{\rho}_{n1} 
&= 
-\cos \alpha_{n} - \cos \theta_{1n}.
\label{eqn:ThetaRhoDynamics}
\end{align}
Note that the branch dynamics \eqref{eqn:ThetaRhoDynamics} are not subject to any additional constraints (aside from the requirement $\rho_{n1}>0$), but the agents in the cycle are subject to a cycle closure constraint on the initial conditions \editKSG{(see \cite{Galloway2013SymmetryStrategies})}, given by\footnote{Note that addition in the indices should be interpreted modulo $n-1$ for agents whose arcs are in ${\cal A}_{cycle}$.}
\begin{align*}
& R\left(\sum_{i=1}^{n} (\pi + \alpha_{i} - \theta_{i,i-1})\right) = \mathds{1}, \\
& \sum_{i=1}^{n} \rho_{i,i+1} R\left(\sum_{j=1}^{i}(\pi + \alpha_j -\theta_{j,j-1})\right) = 0, \; i =1,2,\ldots,n.
\end{align*}

Equilibria for the shape \editKSG{dynamics} (\ref{invarmandyn},\ref{eqn:ThetaRhoDynamics}) correspond to \emph{relative equilibria} for the closed-loop version of the full dynamics \eqref{eqn:particleDynamics}. For a system consisting of a single cycle (without branches), \textit{Proposition~6.1} from \cite{Galloway2013SymmetryStrategies} provided conditions for existence of rectilinear and circling relative equilibria, along with the equilibrium values for the shape variables. In particular, it was demonstrated that equilibrium values for $\theta_{i,i-1}$ for agents on the cycle are given by 
\begin{itemize}
\item Rectilinear Equilibrium: $\theta_{i,i-1} = \pi + \alpha_{i-1}$, and
\item Circling Equilibrium: $\theta_{i,i-1} = \pi - \alpha_{i-1}$.
\end{itemize}

If the cycle agents are in a rectilinear equilibrium state (and therefore $\theta_{21}=\pi + \alpha_1$), it follows from \eqref{eqn:ThetaRhoDynamics} that  $\dot{\theta}_{1n} = \dot{\rho}_{n1} = 0$ if and only if $\theta_{1n} = \pi + \alpha_n$. If the cycle agents are in a circling equilibrium state (and therefore $\theta_{21}=\pi - \alpha_1$), it follows from \eqref{eqn:ThetaRhoDynamics} that  $\dot{\theta}_{1n} = \dot{\rho}_{n1} = 0$ if and only if $\theta_{1n} = \pi - \alpha_n \neq 0$ and $\frac{\rho_{n1}}{\rho_{12}} = \frac{\sin\alpha_n}{\sin\alpha_{1}}>0$.

\begin{proposition}
\label{prop:relativeEquilibrium}
Consider an $n$-agent CB pursuit system with arc set \eqref{eqn:ArcSet} evolving on $M_{CB(\pmb{\alpha)}}$ \editKSG{obeying} (\ref{invarmandyn},\ref{eqn:ThetaRhoDynamics}). 
\begin{enumerate}
	\item A rectilinear relative equilibrium exists if and only if 
	\begin{align}
	\alpha_i = \pi + \alpha_{i-1}, \; i=1,2,\ldots, n-1, 
	\end{align}
	in which case the corresponding equilibrium side lengths $\rho_{ij}$ are arbitrary (i.e., determined by initial conditions) and the equilibrium $\theta_{ji}$ values are given by 
\begin{align}
\label{eqn:equilibriumRectValues}
	\theta_{i,i-1} &= \pi + \alpha_{i-1}, \; i=1,2,\ldots, n-1, \nonumber \\
    \theta_{1n} &= \pi + \alpha_{n}. 
\end{align}
\item A circling relative equilibrium exists if and only if 
\begin{align}
\label{eqn:circPropositionCond1}
\text{i. } &\sin(\alpha_{n}) \sin(\alpha_{1}) > 0, \nonumber \\
&\sin(\alpha_{i}) \sin(\alpha_{i+1}) > 0, \; i=1,2,\ldots, n-2   \\
\label{eqn:circPropositionCond2}
\text{ii. } &\sum_{i=1}^{n-1}\left(\alpha_i \right) = 0,
\end{align}
in which case the corresponding equilibrium angles $\theta_{ji}$ are given by
\begin{equation}
\begin{aligned}
\theta_{i,i-1} &= \pi - \alpha_{i-1}, \; i=1,2,\ldots, n-1,
\\
\theta_{1n} &= \pi - \alpha_{n},
\end{aligned}
\label{eqn:equilibriumCircValues}
\end{equation}
and equilibrium side lengths $\rho_{ij}$ satisfy
\begin{equation}
\begin{aligned}
\frac{\rho_{i,i+1}}{\rho_{i+1,i+2}} 
&= \frac{\sin\alpha_{i}}{\sin\alpha_{i+1}}, \; i=1,2,\ldots, n-1,
\\
\frac{\rho_{n1}}{\rho_{12}} 
&= \frac{\sin\alpha_{n}}{\sin\alpha_{1}}.
\end{aligned}
\label{eqn:equilibriumCircRhoValues}
\end{equation}
\end{enumerate}  
\end{proposition}
\begin{proof}
The conditions and equilibrium values for the agents on the cycle (i.e. agents $1$ through $n-1$) are stated in \textit{Proposition~6.1} in \cite{Galloway2013SymmetryStrategies}. For the branch agent, it follows from the discussion immediately preceding the proposition statement that the value for $\theta_{1n}$ given in \eqref{eqn:equilibriumRectValues} will also set the branch dynamics \eqref{eqn:ThetaRhoDynamics} to zero. From \eqref{shapeVariableDefinitions} one can show (see \cite{Galloway2013SymmetryStrategies}) that on $M_{CB(\pmb{\alpha)}}$ we have
\begin{equation}
{\bf x}_{n} \cdot {\bf x}_{1} = \cos(\pi + \alpha_{n} - \theta_{1n}),
\end{equation}
and therefore substitution of the value for $\theta_{1n}$ given in \eqref{eqn:equilibriumRectValues} results in ${\bf x}_{n} \cdot {\bf x}_{1} =1$, i.e. the tangent vectors are aligned in a rectilinear equilibrium.

For the second case where the cycle agents orbit on a circling equilibrium (per \textit{Proposition~6.1} of \cite{Galloway2013SymmetryStrategies}), the discussion preceding the proposition statement establishes that the values for $\theta_{1n}$ and $\rho_{1n}$ given in \eqref{eqn:circPropositionCond1} and \eqref{eqn:circPropositionCond2} result in setting the branch dynamics \eqref{eqn:ThetaRhoDynamics} to zero. It remains to establish that the branch agent will orbit on the \textit{same} circling equilibrium as agents $1$ through $n-1$. As demonstrated in \cite{Galloway2013SymmetryStrategies}, that circling equilibrium has radius given by $\frac{\rho_{12}}{2\sin\alpha_{1}}$ and circumcenter located at ${\bf r}_{1} + \frac{\rho_{12}}{2\sin\alpha_{1}}{\bf x}_{1}^\perp$, so it suffices to show that ${\bf r}_{n} + \frac{\rho_{12}}{2\sin\alpha_{1}}{\bf x}_{n}^\perp = {\bf r}_{1} + \frac{\rho_{12}}{2\sin\alpha_{1}}{\bf x}_{1}^\perp$. From \eqref{shapeVariableDefinitions}, we have
\begin{align*}
{\bf x}_1 
&= 
R(-\theta_{1n})\left(\frac{{\bf r}_{n}-{\bf r}_1}{\left\|{\bf r}_{n}-{\bf r}_1\right\|}\right) = \frac{1}{\rho_{n1}}R(\alpha_{n}-\pi)({\bf r}_{n}-{\bf r}_1),
\\
{\bf x}_{n} 
&= 
-R(-\alpha_{n})\left(\frac{{\bf r}_{n}-{\bf r}_1}{\left\|{\bf r}_{n}-{\bf r}_1\right\|}\right) = -\frac{1}{\rho_{n1}}R(-\alpha_{n})({\bf r}_{n}-{\bf r}_1).  
\end{align*}
Then, following the calculations in \cite{Galloway2013SymmetryStrategies} we can show that
\begin{align*}
&\left({\bf r}_{n} + \frac{\rho_{12}}{2\sin\alpha_{1}}{\bf x}_{n}^\perp \right) - 
\left({\bf r}_{1} + \frac{\rho_{12}}{2\sin\alpha_{1}}{\bf x}_{1}^\perp \right) \nonumber \\
&\quad = \left[\mathds{1}-\frac{1}{2\sin\alpha_{n}} 2\cos\left(\frac{\pi}{2}-\alpha_{n}\right)\mathds{1}\right]\left({\bf r}_{n}-{\bf r}_1\right) 
\end{align*}
which equals 0, establishing the claim.
\end{proof}
\begin{remark}
Note that the only condition in Proposition \ref{prop:relativeEquilibrium} which depends on the CB parameter for the branch agent (i.e. $\alpha_n$) is \eqref{eqn:circPropositionCond1}, which essentially requires that all the agents follow the same direction of rotation around the circling equilibrium (i.e. clockwise vs. counter-clockwise). In particular, we note that $\alpha_n$ is not required to satisfy the more stringent requirement given by \eqref{eqn:circPropositionCond2}, which means a branch agent can join or leave a circling equilibrium without disturbing (or coordinating with) the other agents. The branch agent's relative position on the circling equilibrium can be modified independently by adjusting the value of $\alpha_n$.  
\end{remark}

\editKSG{Examples of relative equilibria are} depicted in Fig~\ref{fig:threeParticleFigs}a and Fig~\ref{fig:threeParticleFigs}b. Note that in both cases agents 1, 2 and 3 engage in cyclic pursuit on a three-agent relative equilibrium, and agent 4 essentially joins the relative equilibrium, with the value of $\alpha_4$ determining agent \editKSG{4}'s equilibrium position with respect to the other agents.
%
%

%
%
\subsection{Pure shape equilibria}
%
As demonstrated in \cite{Galloway2013SymmetryStrategies}, we can separate the shape dynamics (\ref{invarmandyn},\ref{eqn:ThetaRhoDynamics}) into two parts by using an appropriate re-parametrization of $M_{CB(\pmb{\alpha)}}$ and a subsequent rescaling of the time variable. After this separation, one part describes evolution of the \emph{size}/\emph{scale} of the system, and the other one describes the \emph{pure shape} (i.e. the formation shape up to geometric similarity). Towards this objective we first define the following change of variables
\begin{displaymath}
\lambda \triangleq \ln(\rho_{12}),
\quad \textrm{and} \quad
\tilde{\rho}_{ij} \triangleq \rho_{ij}/\rho_{12} = \rho_{ij}e^{-\lambda}
\end{displaymath}
for every $(i,j) \in \cal{A}$ from \eqref{eqn:ArcSet}. Then, by introducing a time-rescaling defined as
\begin{equation}
\tau = \int_{0}^{t} e^{-\lambda(\sigma)} d\sigma,
\label{eqn:firstTauDefn}
\end{equation}
the shape dynamics (\ref{invarmandyn},\ref{eqn:ThetaRhoDynamics}) can be expressed as
\begin{align}
\lambda^{'} 
&= 
-  \left(\cos\alpha_1 + \cos\theta_{21}\right)
\label{eqn:lambdaPrime}
\\
\theta_{1n}^{'} 
&= 
\frac{1}{\tilde{\rho}_{n1}}\left(\sin\alpha_{n} + \sin\theta_{1n}\right) - \left(\sin\alpha_{1} + \sin\theta_{21}\right)
\label{eqn:thetaPrime}
\\
\tilde{\rho}_{n1}^{'} 
&= 
\tilde{\rho}_{n1}\left(\cos\alpha_1 + \cos\theta_{21}\right) -\left(\cos\alpha_n + \cos\theta_{1n}\right)
\label{eqn:rhoTildePrime}
\\
\theta_{i,i-1}^{'} 
&= 
\frac{1}{\tilde{\rho}_{i-1,i}}\left(\sin\alpha_{i-1} + \sin\theta_{i,i-1}\right)
\nonumber
\\
&\qquad  
- \frac{1}{\tilde{\rho}_{i,i+1}}\left(\sin\alpha_{i} + \sin\theta_{i+1,i}\right)
\label{eqn:thetaPrimeCycle}
\\
\tilde{\rho}_{i,i+1}^{'} 
&= 
\tilde{\rho}_{i,i+1}\left(\cos\alpha_1 + \cos\theta_{21}\right)
\nonumber
\\
& \qquad 
-\left(\cos\alpha_i + \cos\theta_{i+1,i}\right)
\label{eqn:rhoTildePrimeCycle}
\end{align}
for $i=1,\ldots, n-1$, where the \emph{prime} notation denotes differentiation with respect to the rescaled time $\tau$. It is worth noting here that \eqref{eqn:lambdaPrime} describes the evolution of size of the system, while the dynamics \eqref{eqn:thetaPrime}-\eqref{eqn:rhoTildePrimeCycle} are now self-contained and describe the evolution of pure shape. Equilibria for the dynamics \eqref{eqn:thetaPrime}-\eqref{eqn:rhoTildePrimeCycle} are known as \emph{pure shape equilibria}, and correspond to system trajectories which preserve the pure shape of the system while (possibly) evolving in size (e.g. the spiral motion shown in Fig~\ref{fig:threeParticleFigs}c).
%
%
%
%
\begin{figure}[b!]
\centering
$\begin{array}{cc}
\includegraphics[height=.2\textwidth]{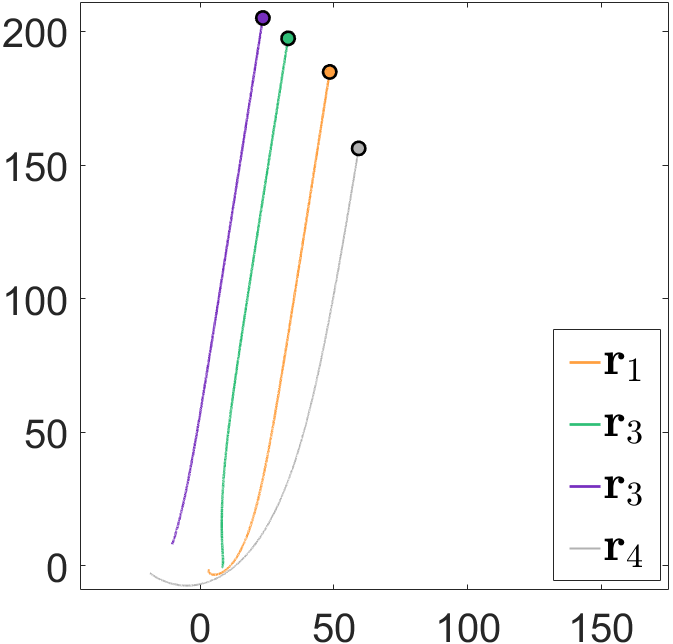}
&
\includegraphics[height=.2\textwidth]{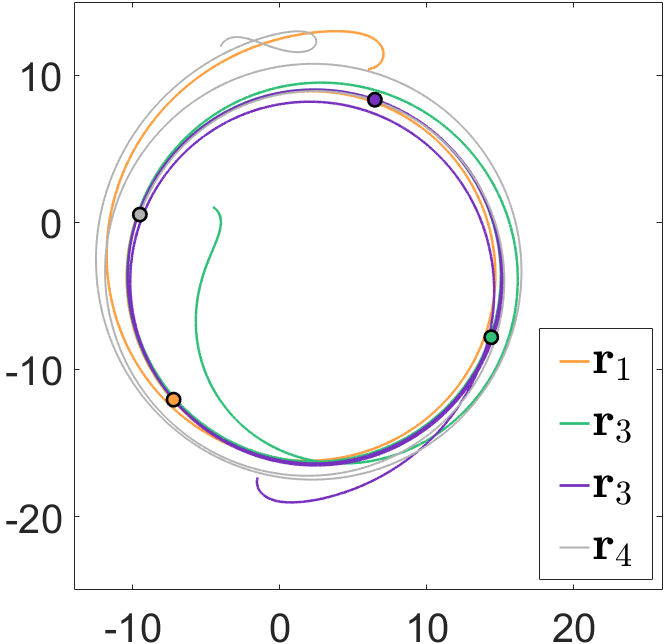}
\\
\textrm{\small{(a): Rectilinear Equilibrium}} & \textrm{\small{(b): Circling Equilibrium}}
\\
\includegraphics[height=.2\textwidth]{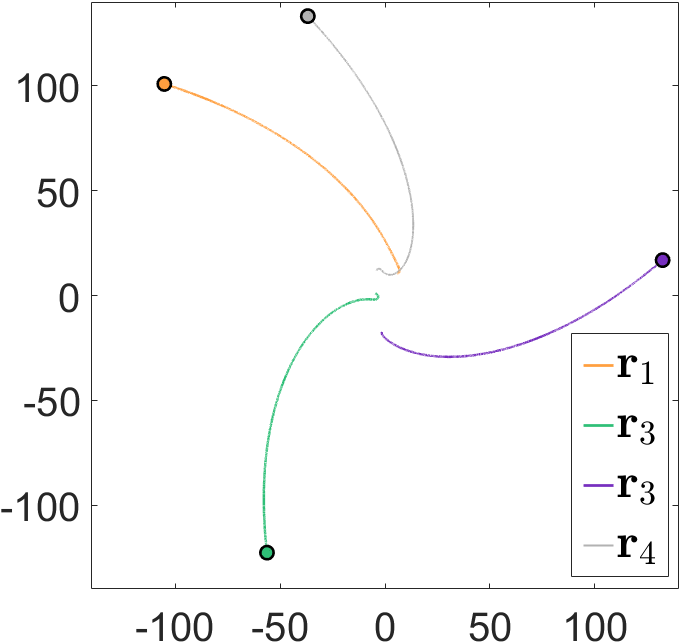}
&
\includegraphics[height=.2\textwidth]{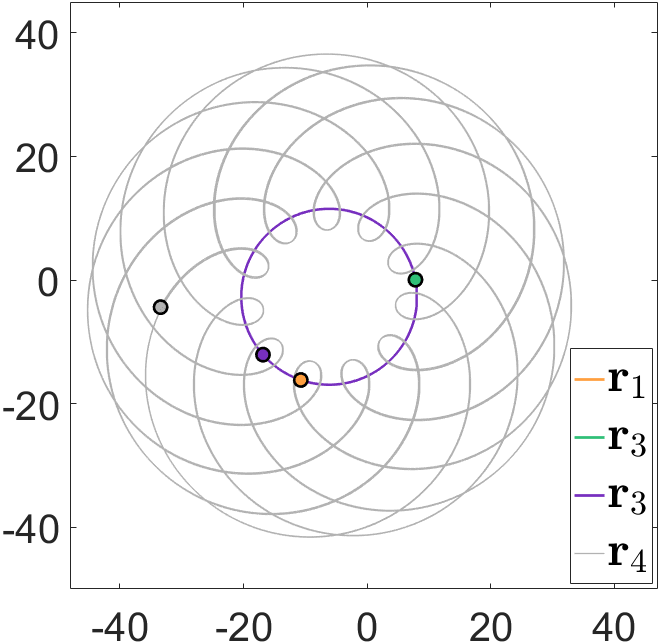}
\\
\textrm{\small{(c): Shape Preserving Spiral}} & \textrm{\small{(d): Periodic Orbit}}
\end{array}$
\caption{\small{This figure illustrates MATLAB simulations for a 4-agent collective, wherein agent 1, 2 and 3 are engaged in a cyclic pursuit, and agent 4 pursues agent 1. Depending on parameter choices the outcome can be different: (a) rectilinear equilibrium ($\alpha_1 = \alpha_2 - \pi = \alpha_3 = \pi/3$, $\alpha_4 = \pi/6$); (b) circling equilibrium ($\alpha_1 =\pi/3$, $\alpha_2=\pi/4$, $\alpha_3 = 5\pi/12$, $\alpha_4 = \pi/6$); (c) shape preserving spirals ($\alpha_1 = \alpha_2 = \alpha_3 = 2\pi/3$, $\alpha_4 = 2\pi/5$); or (d) periodic motion paths ($\alpha_1 =\pi/3$, $\alpha_2=7\pi/12$, $\alpha_3=\pi/12$, $\alpha_4=\pi/2$).}}
\label{fig:threeParticleFigs}
\vspace{-1.5em}
\end{figure}
\begin{proposition}
\label{prop:pureShapeEquilibria}
Consider an $n$-agent CB pursuit system with arc set \eqref{eqn:ArcSet} evolving on $M_{CB(\pmb{\alpha)}}$ according to the shape dynamics \eqref{eqn:lambdaPrime}-\eqref{eqn:rhoTildePrimeCycle}. Pure shape equilibria exist if and only if the conditions of \textit{Proposition~\ref{prop:relativeEquilibrium}} are met, or there exists an \editKSG{integer} $k \in \{0,1,\ldots,n-2 \}$ such that 
\begin{displaymath}
\begin{aligned}
&\sin(\alpha_n - \tau_k)\sin(\alpha_{1}-\tau_k) > 0, 
\\
&\sin(\alpha_i - \tau_k)\sin(\alpha_{i+1}-\tau_k) > 0, \; i=1,2,\ldots,n-2,
\end{aligned}
\end{displaymath}
\editKSG{for $\tau_k \triangleq \left(\sum_{i=1}^{n-1}(\alpha_i/(n-1))\right)-k\pi/(n-1)$, with equilibrium values for the branch agent given by} 
\begin{equation}
\theta_{1n} = \pi - \alpha_n + 2\tau_k, 
\quad \textrm{and,} \quad 
\tilde{\rho}_{n1} = \frac{\sin(\alpha_n - \tau_k)}{\sin(\alpha_1 - \tau_k)}.
\label{eqn:pureShapeEquilibriumValues}
\end{equation}
\end{proposition}
\begin{proof}
A pure shape equilibrium corresponds to trajectories for which \eqref{eqn:thetaPrime}-\eqref{eqn:rhoTildePrimeCycle} are zero. If $\cos \alpha_i + \cos \theta_{i,i+1}=0$ for any $i$, then setting \eqref{eqn:thetaPrime}-\eqref{eqn:rhoTildePrimeCycle} to zero requires $\cos \alpha_1 + \cos \theta_{21}=0$, which corresponds to a relative equilibrium (already covered by \textit{Proposition~\ref{prop:relativeEquilibrium}}). Therefore we proceed by assuming $\cos \alpha_i + \cos \theta_{i,i+1} \neq 0$ for any $i$. 

For the agents in the cycle (i.e. agents $1$ through $n-1$), \textit{Proposition~6.3} of \cite{Galloway2013SymmetryStrategies} provides necessary and sufficient conditions for the existence of pure shape equilibria which are not relative equilibria, \editKSG{which are represented by the second condition in our current proposition.}
\editKSG{The referenced proposition from \cite{Galloway2013SymmetryStrategies}} also allows us to express the equilibrium values for pure shape equilibria on $M_{CB(\pmb{\alpha)}}$, and in particular we have
\begin{align}
\theta_{21}=\pi - \alpha_1 +2 \tau_k.
\end{align}
By substitution into \eqref{eqn:rhoTildePrime}, we have $\tilde{\rho}_{n1}^{'} = 0$ if and only if
\begin{align}
\label{eqn:pureShapeCondition1}
\tilde{\rho}_{n1} = \frac{\cos \alpha_n +\cos \theta_{1n} }{\cos \alpha_1 +\cos \theta_{21} } 
= \frac{\cos \alpha_n +\cos \theta_{1n} }{-2\sin(\tau_k)\sin (\alpha_1 -\tau_k) }.
\end{align}
If we further assume that $\sin \alpha_i +\sin \theta_{i+1,i} \neq 0$ for every $i$, from \eqref{eqn:thetaPrime} we observe that $\theta_{1n}^{'} = 0$ if and only if
\begin{align}
\label{eqn:pureShapeCondition2}
\tilde{\rho}_{n1} = \frac{\sin \alpha_n +\sin \theta_{1n} }{\sin \alpha_1 +\sin \theta_{21} } 
= \frac{\sin \alpha_n +\sin \theta_{1n} }{2\cos(\tau_k)\sin (\alpha_1 -\tau_k) }.
\end{align} 
Equating \eqref{eqn:pureShapeCondition1} with \eqref{eqn:pureShapeCondition2} and employing appropriate trigonometric identities, \editKSG{we have the equilibrium values for the branch shape dynamics given by} \eqref{eqn:pureShapeEquilibriumValues}, with the requirement that $\sin(\alpha_n - \tau_k)\sin(\alpha_1 - \tau_k) > 0$.

Lastly, we address the case where $\sin \alpha_i +\sin \theta_{i+1,i} = 0$ for every $i$. (Note that setting \eqref{eqn:thetaPrimeCycle} to zero precludes the situation where $\sin \alpha_i +\sin \theta_{i+1,i} = 0$ for only some $i$.) Since $\cos \alpha_i + \cos \theta_{i,i+1} \neq 0$, the only possibility is that $\theta_{i+1,i} = -\alpha_i$, for every $i$. This situation is addressed in the supplementary material of \cite{Galloway2013SymmetryStrategies}, where it is shown that this is a particular case of \eqref{eqn:pureShapeEquilibriumValues} for which $\tau_k = \pm \pi/2$. 
\end{proof}
%

%
%
\subsection{A Special Case: Periodic Orbits in the Pure Shape Space}
%
Now we focus on a special case, and show existence of interesting behavior (as shown in Fig~\ref{fig:threeParticleFigs}d) on the pure shape space. If the conditions \eqref{eqn:circPropositionCond1}-\eqref{eqn:circPropositionCond2} hold true, i.e. a circling equilibrium exists, and $\alpha_n = \pi/2$, we note that the dynamics \eqref{eqn:thetaPrime}-\eqref{eqn:rhoTildePrime} simplify to 
\begin{equation}
\begin{aligned}
\theta_{1n}^{'} &= \frac{1}{\tilde{\rho}_{n1}}\left(1 + \sin\theta_{1n}\right) - 2\sin\alpha_1, \\
\tilde{\rho}_{n1}^{'} &= -\cos\theta_{1n},
\end{aligned}
\label{eqn:halfPiDynamics}
\end{equation}
defined on the set $\{(\theta_{1n},\tilde{\rho}_{n1}): -\pi < \theta_{1n} \leq \pi, \tilde{\rho}_{n1}>0 \}$. Whenever $\sin\alpha_1 > 0$, this dynamics \eqref{eqn:halfPiDynamics} has a unique equilibrium point at $(\theta_{1n},\tilde{\rho}_{n1}) = (\pi/2,1/\sin\alpha_1)$, and one can easily show that the eigenvalues of the linearized dynamics are given by $\lambda = \pm j \sqrt{2}\sin(\alpha_1)$. This suggests the possibility of periodic orbits, and in what follows, we will demonstrate that the system does in fact admit periodic orbits, even when $\sin\alpha_1 < 0$.

Then, by adopting an approach similar to (\cite{Mischiati2012894, Udit_CDC16}), we show that the solutions of \eqref{eqn:halfPiDynamics} constitute periodic orbits in the underlying state space. Towards this objective, we first introduce some relevant definitions and the following theorem due to G. D. Birkhoff \cite{Birkhoff_Original}.

%
%
%
%

%
%
\begin{definition}[Involution]
A diffeomorphism $F:\mathcal{M} \rightarrow \mathcal{M}$ defined on a manifold $\mathcal{M}$ is an \textit{involution} if $F \neq id_{\mathcal{M}}$, the identity diffeomorphism, and $F\big(F(m)\big) = m, \; \forall m \in \mathcal{M}$.
\end{definition}
\begin{definition}[F-reversibility]
A vector field $\mathfrak{X}$ defined on a manifold $\mathcal{M}$ is said to be $F$-reversible if there exists an involution $F$ such that its pushforward $F_*\mathfrak{X} = -\mathfrak{X}$.
\end{definition}
\begin{theorem}[G. D. Birkhoff \cite{Birkhoff_Original}]
Let $\mathfrak{X}$ be a $F$-reversible vector field on $M$ and $\Sigma_F$ denote the fixed-point set of the reverser $F$. If an orbit of $\mathfrak{X}$ through a point of $\Sigma_F$ intersects $\Sigma_F$ at another point, then it is periodic.
\end{theorem}
\begin{lemma}
The vector field defined by \eqref{eqn:halfPiDynamics} is $F$-reversible, with the reverser $F(\theta_{1n},\tilde{\rho}_{n1}) = (\pi - \theta_{1n},\tilde{\rho}_{n1})$.
\end{lemma}
\begin{proof}
It is straightforward to show that the diffeomorphism $F(\theta_{1n},\tilde{\rho}_{n1}) = (\pi - \theta_{1n},\tilde{\rho}_{n1})$ is indeed an involution. The associated pushforward maps the vector field $\mathfrak{X}$ defined by \eqref{eqn:halfPiDynamics} into the vector field:
\begin{align*}
&F_*\mathfrak{X}(\theta_{1n},\tilde{\rho}_{n1})
\\
&\quad=
(DF)_{F^{-1}(\theta_{1n},\tilde{\rho}_{n1})} \cdot \mathfrak{X} \big( F^{-1}(\theta_{1n},\tilde{\rho}_{n1}) \big)
\\
&\quad=
\left[\begin{array}{rr}
-1 & 0 \\ 0 & 1
\end{array}\right]
\left[\begin{array}{l}
\frac{1}{\tilde{\rho}_{n1}}\left(1 + \sin(\pi-\theta_{1n})\right) - 2\sin\alpha_1 
\\
-\cos(\pi-\theta_{1n})
\end{array}\right]
\\
&\quad=
-\mathfrak{X}(\theta_{1n},\tilde{\rho}_{n1}).
\end{align*}
Hence the vector field defined by \eqref{eqn:halfPiDynamics} is $F$-reversible.
\end{proof}

Now we state our main result.
\begin{theorem}
Every solution trajectory of \eqref{eqn:halfPiDynamics}:\\
\indent (a) has a conserved quantity:
\begin{equation}
E(\theta_{1n},\tilde{\rho}_{n1}) 
\triangleq 
\tilde{\rho}_{n1}(1 + \sin\theta_{1n}) -\tilde{\rho}_{n1}^2 \sin\alpha_1,
\label{eqn:halfPiDynamics_CONSRVD_qty}
\end{equation}
\indent (b) is a periodic orbit.
\end{theorem}
\begin{proof}
\editKSG{Our proof follows the line of thinking from \cite{Mischiati2012894, Udit_CDC16}, as follows.}

(a) A direct calculation of the derivative with respect to the rescaled time $\tau$ would yield
\begin{align*}
\frac{dE}{d\tau} = \frac{\partial E}{\partial \theta_{1n}}\cdot\theta_{1n}^{'} + \frac{\partial E}{\partial \tilde{\rho}_{n1}}\cdot\tilde{\rho}_{n1}^{'} = 0,
\end{align*}
and hence $E(\theta_{1n},\tilde{\rho}_{n1})$ is conserved along the trajectories.

(b) As $\{\pi/2,-\pi/2\}\times\mathds{R}^+$ constitutes the fixed-point set $\Sigma_F$ of the reverser $F$, we can complete the proof by showing that any solution trajectory of \eqref{eqn:halfPiDynamics} through a point on the $\theta_{1n} = \pm\pi/2$ line hits the $\theta_{1n} = \pm\pi/2$ line again.
\begin{figure}[t!]
\centering
$\begin{array}{ccc}
\includegraphics[height=.165\textwidth]{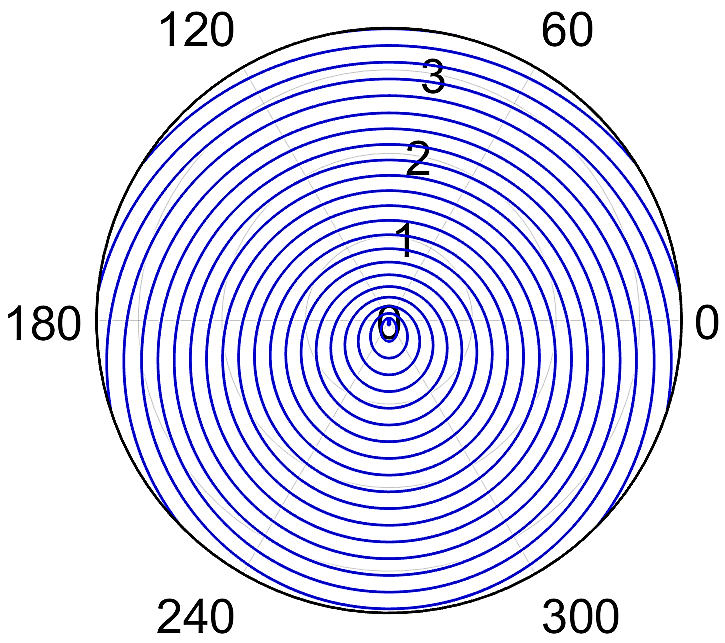}
&&
\includegraphics[height=.165\textwidth]{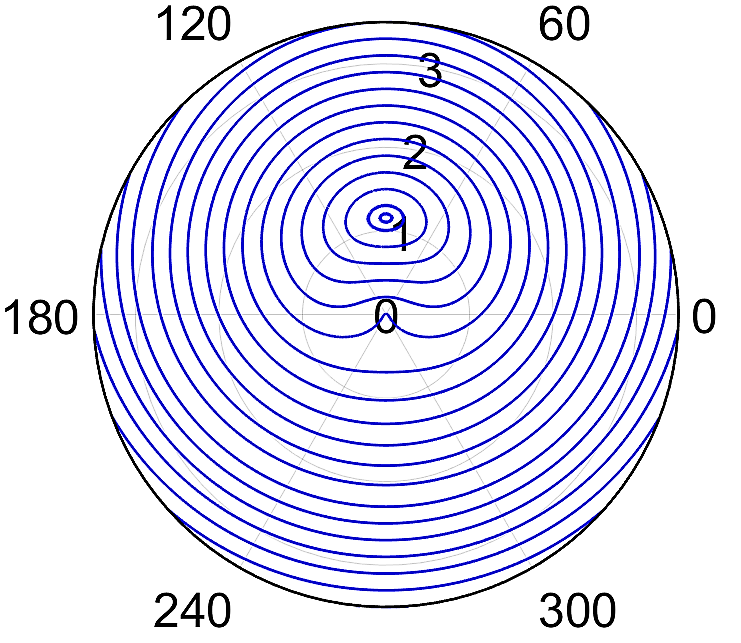}
\\
\small{\textrm{(a) $\alpha_1 = -\pi/3$}} && \small{\textrm{(b) $\alpha_1 = \pi/3$}}
\end{array}$
\caption{\small{Phase portrait of the $(\theta_{1n},\tilde{\rho}_{n1})$ dynamics for the scenario where \eqref{eqn:circPropositionCond1}-\eqref{eqn:circPropositionCond2} hold true and $\alpha_n = \pi/2$.}}
\label{PhasePortrait}
\vspace{-1.5em}
\end{figure}

First we consider the case when $\alpha_1 \in (-\pi,0)$, i.e. $\sin\alpha_1<0$. Then $\theta_{1n}^{'}$ is always positive, which in turn ensures that any trajectory originating from the $\theta_{1n} = \pm\pi/2$ line will travel counter-clockwise until it intersects the $\theta_{1n} = \mp\pi/2$ line again, when $\theta_{1n}$ has been incremented by angle $\pi$ (as shown in Fig~\ref{PhasePortrait}a).
\begin{wrapfigure}{L}{0.21\textwidth}
\centering
\includegraphics[width=.19\textwidth]{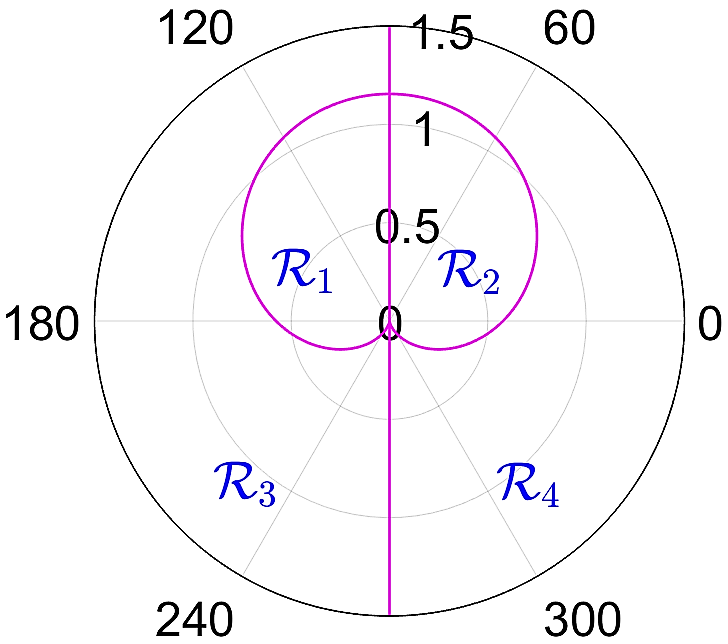}
\caption{\small{Nullclines of $(\theta_{1n},\tilde{\rho}_{n1})$ dynamics with $\alpha_1 =\pi/3$.}}
\label{NullClines_of_47}
\vspace{-.75em}
\end{wrapfigure}

Next we consider the case when $\alpha_1 \in (0,\pi)$, i.e. $\sin\alpha_1>0$. In this case we have an equilibrium point at $(\pi/2,1/\sin\alpha_1)$, and the nullclines (as shown in Fig~\ref{NullClines_of_47}) are given by $\tilde{\rho}_{n1} = \frac{1+\sin\theta_{1n}}{2\sin\alpha_1}$ (for $\theta_{1n}^{'} =0$) and the $\theta_{1n} = \pm\pi/2$ line (for $\tilde{\rho}_{n1}^{'}=0$). Then we partition the state-space into 4 regions $\mathcal{R}_1$ (where $\theta_{1n}^{'}> 0$ and $\tilde{\rho}_{n1}^{'}> 0$), $\mathcal{R}_2$ (where $\theta_{1n}^{'}> 0$ and $\tilde{\rho}_{n1}^{'}< 0$), $\mathcal{R}_3$ (where $\theta_{1n}^{'}< 0$ and $\tilde{\rho}_{n1}^{'}> 0$) and $\mathcal{R}_4$ (where $\theta_{1n}^{'}< 0$ and $\tilde{\rho}_{n1}^{'}< 0$). 

Clearly any trajectory originating on the boundary between
$\mathcal{R}_1$ and $\mathcal{R}_2$ will enter the partition $\mathcal{R}_1$ because $\theta_{1n}^{'}< 0$ and $\tilde{\rho}_{n1}^{'}= 0$ on the boundary. Then it is straightforward to show that these trajectories will later enter into $\mathcal{R}_3$ as $\tilde{\rho}_{n1}^{'}> 0$ and $\theta_{1n}^{'}= 0$ on the boundary between $\mathcal{R}_2$ and $\mathcal{R}_3$. This also implies that trajectories cannot enter into into $\mathcal{R}_2$ from $\mathcal{R}_3$. On the other hand, as $\theta_{1n}^{'}< 0$ within $\mathcal{R}_3$ it can be shown that trajectories can leave this partition through the intersection of the $\theta_{1n} = \pi/2$ line and the boundary between $\mathcal{R}_3$ and $\mathcal{R}_4$. This allows us to conclude that any solution trajectory of \eqref{eqn:halfPiDynamics} crosses the $\theta_{1n} = \pm\pi/2$ line twice.
\end{proof}
%
%
%

%
%
%
\section{Stability Analysis for A 3-Agent System}
\label{eqn:stabilityAnalysis}
%
Now we narrow our focus to \editKSG{a special case (i.e., $n=3$),} wherein the pursuit graph is defined as $\mathcal{G} = (\{1,2,3\},\{(1,2),(2,1),(3,1)\})$, i.e. agent 1 and 2 are engaged in a mutual pursuit and agent 3 pursues agent 1. \editKG{As demonstrated in \cite{Galloway2013SymmetryStrategies}, the dynamics for mutual pursuit on the CB manifold simplify to $\dot{\rho}_{12} = -\cos(\alpha_1)-\cos(\alpha_2)$. Therefore, we can assess the stability of the special solutions described in \textit{Proposition~\ref{prop:relativeEquilibrium}} and \textit{Proposition~\ref{prop:pureShapeEquilibria}} completely in terms of the branch dynamics. We proceed by} linearization of the dynamics \eqref{eqn:thetaPrime}-\eqref{eqn:rhoTildePrime}.

By defining $x \triangleq (\theta_{13},\tilde{\rho}_{31})^{T}$ and letting $f(x)$ denote the dynamics \eqref{eqn:thetaPrime}-\eqref{eqn:rhoTildePrime}, it is straightforward to show that the Jacobian for these dynamics is given by
\begin{align}
\label{eqn:generalJacobian}
	\frac{\partial f}{\partial x} = \left(\begin{array}{cc} 
	\frac{\cos \theta_{13}}{\tilde{\rho}_{31}} &  -\frac{\sin \alpha_3 + \sin \theta_{13}}{\tilde{\rho}_{31}^{2}}\\ \sin \theta_{13}  & \cos(\alpha_1) + \cos(\alpha_2) \end{array}\right).
\end{align}

\begin{proposition}
\label{prop:stability} Consider a 3-agent system wherein  agents 1 and 2 are engaged in a mutual pursuit, and agent 3 pursues agent 1 with CB bearing angle $\alpha_3$. The following provides a necessary and sufficient condition for local stability of the circling and pure shape equilibria:
\\
\indent (a) The circling equilibria of \textit{Proposition~\ref{prop:relativeEquilibrium}} are stable if $\cos \alpha_3 > 0$ and unstable if $\cos \alpha_3 <0$.
\\
\indent (b) The pure shape equilibria of \textit{Proposition~\ref{prop:pureShapeEquilibria}} are stable if $2\cos(\alpha_1 + \alpha_2-\alpha_3) + \cos(\alpha_3) < 0$ and unstable if $2\cos(\alpha_1 + \alpha_2-\alpha_3) + \cos(\alpha_3) > 0$. 
\end{proposition}
\begin{proof}
\editKSG{For mutual pursuit pure shape equilibria}, it can be shown that $\tau_k$ from \textit{Proposition~\ref{prop:pureShapeEquilibria}} must be $\tau_k = \frac{\alpha_{1} + \alpha_{2}}{2} - \frac{\pi}{2}$. Thus from \eqref{eqn:pureShapeEquilibriumValues}, we have equilibrium values for the branch agent given by
\begin{equation}
{\theta}^*_{13} = \alpha_1 + \alpha_2 - \alpha_3, \quad
{\tilde{\rho}}^*_{31} = \frac{\cos\left(\alpha_3 - \frac{\alpha_1 + \alpha_2}{2}\right)}{\cos\left(\alpha_1 - \frac{\alpha_1 + \alpha_2}{2}\right)}.
\label{EQ_PURE_SHAPE_3AGENT}
\end{equation}
Then letting $\beta \triangleq \frac{\alpha_{1} + \alpha_{2}}{2}$ and substituting \eqref{EQ_PURE_SHAPE_3AGENT} into \eqref{eqn:generalJacobian}, we can express the Jacobian $\frac{\partial f}{\partial x}\Big|_{PS}$ as
\begin{displaymath}
\left(\begin{array}{cc} 
\frac{\cos(2\beta -\alpha_3) \cos(\alpha_1 - \beta)}{\cos(\alpha_3 - \beta)} &  -\frac{\cos^{2}(\alpha_1 - \beta)[\sin \alpha_3 + \sin(2\beta -\alpha_3)]}{\cos^{2}(\alpha_3 - \beta)}\\ \sin(2\beta -\alpha_3)  & 2\cos(\beta)\cos(\alpha_1 - \beta) 
\end{array}\right),
\end{displaymath}
and the associated determinant is
\begin{equation}
\det\left(\frac{\partial f}{\partial x}\Big|_{PS} \right) 
=
2\cos^{2}(\alpha_1 - \beta),
\end{equation}
which is strictly positive since \textit{Proposition~\ref{prop:pureShapeEquilibria}} requires $\cos(\alpha_1 - \beta) \neq 0$. Since the eigenvalues are given (in terms of the trace and determinant) by
\begin{displaymath}
\textrm{\small{$\lambda = \frac{1}{2}\left(\text{tr}\left(\frac{\partial f}{\partial x}\Big|_{PS}\right) \pm \sqrt{\text{tr}^{2}\left(\frac{\partial f}{\partial x}\Big|_{PS}\right) - 4\text{det}\left(\frac{\partial f}{\partial x}\Big|_{PS}\right)}\right)$}},
\end{displaymath}
and the determinant is strictly positive, it holds that the real part of the eigenvalues has the same sign as the trace, i.e.
\begin{equation}
\sgn(Re(\lambda))
=
\sgn\Bigl(2\cos(\alpha_1 + \alpha_2-\alpha_3) + \cos(\alpha_3) \Bigr),
\label{eqn:eigenvalueSignPS}
\end{equation}
where we have used the fact that \textit{Proposition~\ref{prop:pureShapeEquilibria}} requires $\cos\left(\alpha_1 - \beta\right)\cos\left(\alpha_3 - \beta\right) > 0$. Note that in terms of the equilibrium values given in \eqref{EQ_PURE_SHAPE_3AGENT}, we can also express \eqref{eqn:eigenvalueSignPS} as $\sgn(Re(\lambda))=\sgn\left(2\cos{\theta}^*_{13} + \cos\alpha_3 \right)$.

Finally, note that circling equlibria can be viewed as a special case of pure shape equilibria for which $\tau_k = 0$ or $\pi$ (i.e. $\beta = \pm \frac{\pi}{2}$). Therefore \eqref{eqn:eigenvalueSignPS} applies, and since we have $\alpha_2 = \pi - \alpha_1$ on a circling equilibrium, it follows that $2\cos(\alpha_1 + \alpha_2-\alpha_3) + \cos(\alpha_3) = -\cos\alpha_3$. 
\end{proof}

\begin{remark}
If we evaluate \eqref{eqn:generalJacobian} at the rectilinear equilibrium from \textit{Proposition~\ref{prop:relativeEquilibrium}}, the Jacobian can be expressed as
\begin{align}
	\frac{\partial f}{\partial x}\Big|_{rect} = 	\left(\begin{array}{cc} 
	-\frac{\cos \alpha_3}{{\tilde{\rho}}^*_{31}} &  0\\ -\sin \alpha_{3}  & 0 \end{array}\right).
\end{align}
From this it follows that the corresponding eigenvalues are given by $\lambda = -\frac{\cos \alpha_3}{{\tilde{\rho}}^*_{31}}$ and a zero eigenvalue resulting from the fact that there exists not a single equilibrium point but a whole continuum of equilibria. Numerical simulations and phase portrait analysis suggest that the continuum of rectilinear equilibria is attractive for $\cos \alpha_3>0$, a conjecture which we intend to explore further in future work.
\end{remark}
%
%
%
%
%
%
%
%

%
%
%
\section{Future work}
\label{sec:SecondBranchDynamics}
%
\editKG{There are two immediate extensions of our proposed framework which we intend to analyze more completely in future work. The first extension considers multiple agents pursuing agent $1$ using a CB pursuit law, i.e. multiple branches off one cycle agent. Clearly the branch dynamics will be independent of one another, and it follows that the results of \emph{Propositions~\ref{prop:relativeEquilibrium}}, \emph{\ref{prop:pureShapeEquilibria}}, and \emph{\ref{prop:stability}} could be extended to this multiple branch case. The other extension involves a single open chain of agents in CB pursuit with its head pursuing agent $1$, i.e. a multi-tiered branch attached to a single cycle agent. In this case, branch agents are influenced both by the cycle agents and by any other branch agent which is closer to the cycle. As will be demonstrated in future work, it follows that results analogous to those in {\it Propositions \ref{prop:relativeEquilibrium}, \ref{prop:pureShapeEquilibria}}, and {\it \ref{prop:stability}} can be extended to this case.}
\bibliographystyle{IEEEtran}
\bibliography{Biswa_Refs,GallowayMendeley}  

\begin{thebibliography}{10}
\providecommand{\url}[1]{#1}
\csname url@samestyle\endcsname
\providecommand{\newblock}{\relax}
\providecommand{\bibinfo}[2]{#2}
\providecommand{\BIBentrySTDinterwordspacing}{\spaceskip=0pt\relax}
\providecommand{\BIBentryALTinterwordstretchfactor}{4}
\providecommand{\BIBentryALTinterwordspacing}{\spaceskip=\fontdimen2\font plus
\BIBentryALTinterwordstretchfactor\fontdimen3\font minus
  \fontdimen4\font\relax}
\providecommand{\BIBforeignlanguage}[2]{{%
\expandafter\ifx\csname l@#1\endcsname\relax
\typeout{** WARNING: IEEEtran.bst: No hyphenation pattern has been}%
\typeout{** loaded for the language `#1'. Using the pattern for}%
\typeout{** the default language instead.}%
\else
\language=\csname l@#1\endcsname
\fi
#2}}
\providecommand{\BIBdecl}{\relax}
\BIBdecl

\bibitem{Justh_PSK_SCL04}
E.~W. Justh and P.~S. Krishnaprasad, ``Equilibria and steering laws for planar
  formations,'' \emph{Systems \& Control Letters}, vol.~52, no.~1, pp. 25 --
  38, 2004.

\bibitem{5160735}
J.~L. Ramirez, M.~Pavone, E.~Frazzoli, and D.~Miller, ``Distributed control of
  spacecraft formation via cyclic pursuit: Theory and experiments,'' in
  \emph{Proceedings of the American Control Conference (ACC)}, 2009, pp. 4811
  -- 4817.

\bibitem{CollectiveMot_PrstEscp_Couzin}
P.~Romanczuk, I.~D. Couzin, and L.~Schimansky-Geier, ``Collective motion due to
  individual escape and pursuit response,'' \emph{Physical Review Letters},
  vol. 102, no.~1, p. 010602, 2009.

\bibitem{Marshall2004}
J.~A. Marshall, M.~E. Broucke, and B.~A. Francis, ``Formations of vehicles in
  cyclic pursuit,'' \emph{IEEE Transactions on Automatic Control}, vol.~49,
  no.~11, pp. 1963--1974, 2004.

\bibitem{Marshall20063}
------, ``Pursuit formations of unicycles,'' \emph{Automatica}, vol.~42, no.~1,
  pp. 3 -- 12, 2006.

\bibitem{Sinha20071954}
A.~Sinha and D.~Ghose, ``Generalization of nonlinear cyclic pursuit,''
  \emph{Automatica}, vol.~43, no.~11, pp. 1954 -- 1960, 2007.

\bibitem{Galloway2013SymmetryStrategies}
K.~S. Galloway, E.~W. Justh, P.~S. Krishnaprasad, and P.~R.~S. A, ``{Symmetry
  and reduction in collectives : cyclic pursuit strategies},''
  \emph{Proceedings of the Royal Society A: Mathematical, Physical and
  Engineering Sciences}, vol. 469, no. 2158, pp. 1--23, 2013.

\bibitem{Galloway2016SymmetryPursuit}
------, ``{Symmetry and reduction in collectives : low-dimensional cyclic
  pursuit},'' \emph{Proceedings of the Royal Society A: Mathematical, Physical
  and Engineering Sciences}, vol. 472, pp. 1--21, 2016.

\bibitem{Galloway2015StationPursuit}
K.~S. Galloway and B.~Dey, ``Station keeping through beacon-referenced cyclic
  pursuit,'' in \emph{Proceedings of the American Control Conference (ACC)},
  2015, pp. 4765--4770.

\bibitem{KSG_BD_2016_ACC}
------, ``Stability and pure shape equilibria for beacon-referenced cyclic
  pursuit,'' in \emph{Proceedings of the American Control Conference (ACC)},
  2016, pp. 161--166.

\bibitem{Mallik_Sinha_ECC_15}
G.~R. Mallik, S.~Daingade, and A.~Sinha, ``Consensus based deviated cyclic
  pursuit for target tracking applications,'' in \emph{Proceedings of the
  European Control Conference (ECC)}, 2015, pp. 1718 -- 1723.

\bibitem{Daingade2016AImplementation}
S.~Daingade, A.~Sinha, A.~Vivek~Borkar, and H.~Arya, ``{A variant of cyclic
  pursuit for target tracking applications: theory and implementation},''
  \emph{Autonomous Robots}, vol.~40, pp. 669--686, 2016.

\bibitem{Galloway2016StateSystems}
K.~Galloway and L.~DeVries, ``State observation and parameter estimation in
  cyclic pursuit systems,'' in \emph{Proceedings of the IEEE Conference on
  Decision and Control (CDC)}, 2016, pp. 1781--1786.

\bibitem{Wei2009PursuitGame}
E.~Wei, E.~W. Justh, and P.~S. Krishnaprasad, ``{Pursuit and an evolutionary
  game},'' \emph{Proceedings of the Royal Society A: Mathematical, Physical and
  Engineering Sciences}, vol. 465, no. 2105, pp. 1539--1559, 2009.

\bibitem{Mischiati2012894}
M.~Mischiati and P.~S. Krishnaprasad, ``The dynamics of mutual motion
  camouflage,'' \emph{Systems \& Control Letters}, vol.~61, no.~9, pp. 894 --
  903, 2012.

\bibitem{Udit_CDC16}
U.~Halder, B.~Schlotfeldt, and P.~S. Krishnaprasad, ``Steering for beacon
  pursuit under limited sensing,'' in \emph{Proceedings of the 55th Conference
  on Decision and Control (CDC)}, 2016, pp. 3848--3855.

\bibitem{Birkhoff_Original}
G.~D. Birkhoff, ``The restricted problem of three bodies,'' \emph{Rendiconti
  del Circolo Matematico di Palermo}, vol.~39, no.~1, pp. 265--334, 1915.

\end{thebibliography}
%
%
%
\end{document}